\documentclass[11pt]{article}

\usepackage[a4paper,margin=29mm]{geometry}
\usepackage{amsmath,amssymb,amsthm,mathtools}
\usepackage{microtype}
\usepackage{booktabs}
\usepackage{authblk}
\usepackage[colorlinks=true,citecolor=blue,linkcolor=blue]{hyperref}

\newtheorem{theorem}{Theorem}[section]
\newtheorem{proposition}[theorem]{Proposition}
\newtheorem{lemma}[theorem]{Lemma}
\newtheorem{corollary}[theorem]{Corollary}

\theoremstyle{definition}

\theoremstyle{remark}
\newtheorem{remark}[theorem]{Remark}

\newcommand{\F}{\mathbb F}
\newcommand{\A}{\mathbb A}
\newcommand{\PP}{\mathbb P}
\newcommand{\BCH}{\operatorname{BCH}}
\newcommand{\Supp}{\operatorname{Supp}}
\newcommand{\Span}{\operatorname{span}}
\newcommand{\Frob}{\operatorname{Frob}}

\title{Asymptotic Bounds on Generalized Covering Radii of Binary Primitive BCH Codes}
\author[1]{Maosheng Xiong\thanks{\texttt{mamsxiong@ust.hk}}}
\author[1]{Chi Hoi Yip\thanks{\texttt{machyip@ust.hk}}}
\author[2]{Ferdinando Zullo\thanks{\texttt{ferdinando.zullo@unicampania.it}}}

\affil[1]{Department of Mathematics, The Hong Kong University of Science and Technology, Hong Kong, P. R. China}

\affil[2]{Dipartimento di Matematica e Fisica, Universit\`a degli Studi della Campania ``Luigi Vanvitelli'', Caserta, Italy}
\date{}

\begin{document}

\maketitle

\begin{abstract}
Fix integers $e\ge2$ and $r\ge1$. In this paper we study the $r$-th generalized covering radius $\rho_r\left(\BCH(e,m)\right)$ of the binary primitive $e$-error-correcting 
BCH code $\BCH(e,m)$. By using an algebraic-geometric reformulation of the covering problem together with an explicit Lang--Weil estimate, we prove that 
\[
\rho_r\bigl(\BCH(e,m)\bigr)
\le
(r+1)e-1
\]
for all sufficiently large $m$. For $e\ge7$, this improves a recent result of
Belinsky--Zabokritskiy. Our proof gives a substantially simpler
geometric approach to this upper bound. In particular it implies that 
\[\rho_2\bigl(\BCH(e,m)\bigr)=3e-1\]
for all sufficiently large $m$. Previously it was only known that 
\[\rho_2\bigl(\BCH(e,m)\bigr) \in \left\{3e-1,3e\right\}\]
for all sufficiently large $m$. 
\end{abstract}



\section{Introduction}

The covering radius is a basic geometric parameter of a code, measuring
how well its codewords fill the ambient Hamming space; see, for example,
\cite{CohenKarpovskyMattsonSchatz}. For a binary linear code
$C\subseteq\F_2^n$, its covering radius is
\[
\rho(C)
:=
\max_{v\in\F_2^n}\min_{c\in C}d(v,c),
\]
where $d$ denotes the Hamming distance.

Generalized covering radii were introduced by Elimelech, Firer, and
Schwartz \cite{ElimelechFirerSchwartz} as a higher-order extension of
the covering radius, motivated by database linear querying, including
private information retrieval. Roughly speaking, the $r$-th
generalized covering radius measures the number of coordinate positions
needed to cover $r$ target vectors simultaneously, using a common set
of positions. Elimelech, Firer, and Schwartz gave equivalent
combinatorial, geometric, and algebraic formulations and established
connections with generalized Hamming weights. Subsequent work has
considered, among other topics, Reed--Muller codes
\cite{ElimelechWeiSchwartz} and asymptotic rate questions
\cite{ElimelechSchwartz}. A recent finite-geometric reformulation in
terms of $(\rho,t)$-saturating sets was developed in \cite{AlfaranoMarinoNeriTrombetti}.

For the present paper, the most convenient formulation is in terms of
syndromes. Let $C\subseteq\F_2^n$ be a binary linear code of codimension
$k$, with a full-rank parity-check matrix
\[
H=(h_1,\ldots,h_n),
\qquad
h_i\in\F_2^k.
\]
The $r$-th generalized covering radius $\rho_r(C)$ is the smallest
integer $R$ such that, for every ordered $r$-tuple of syndromes
$s^{(1)},\ldots,s^{(r)}\in\F_2^k$, there exists a set
$I\subseteq\{1,\ldots,n\}$ with $|I|\le R$ satisfying
\[
s^{(a)}
\in
\Span_{\F_2}\{h_i:i\in I\},
\qquad
1\le a\le r.
\]
In particular, $\rho_1(C)=\rho(C)$. All spans in this paper are over
$\F_2$.

Fix $e\ge2$ and put $q=2^m$ and $F=\F_q$, with
$F^\times=F\setminus\{0\}$. We denote by $\BCH(e,m)$ the binary
primitive narrow-sense BCH code of length $q-1$ and designed distance
$2e+1$. In the usual full-rank range, its parity-check columns can be
indexed by $x\in F^\times$ and written as
\begin{equation}\label{eq:column}
h_e(x)
=
\begin{pmatrix}
x\\
x^3\\
x^5\\
\vdots\\
x^{2e-1}
\end{pmatrix}
\in F^e.
\end{equation}
After fixing an $\F_2$-basis of $F$, this is interpreted as a binary
column of height $em$. We refer to the field element $x$ indexing
$h_e(x)$ as its \emph{locator}. For convenience, we also set
$h_e(0)=0$; the value $x=0$ is used only as an auxiliary locator and
does not correspond to a coordinate of the BCH code.

Throughout the paper we assume that the odd binary cyclotomic cosets
corresponding to
\[
1,3,\ldots,2e-1
\]
are pairwise distinct and have size $m$. For fixed $e$ this holds once
$m$ is sufficiently large. A standard sufficient condition is
\[
2e-1\le 2^{\lceil m/2\rceil};
\]
see, for example, \cite[Lemmas~8 and~9]{AlyKlappeneckerSarvepalli} and
\cite[Section~II-D]{XiongYip}. Our results concern the asymptotic regime
in which $e$ and $r$ are fixed and $m$ tends to infinity.

The case $r=1$ is classical, so our main interest is in the higher
generalized covering radii. Our main upper bound is the following.

\begin{theorem}\label{thm:upper}
Fix $e\ge2$ and $r\ge1$. There exists $m_0(e,r)$ such that, for every
$m\ge m_0(e,r)$,
\[
\rho_r\bigl(\BCH(e,m)\bigr)\le re+e-1.
\]
\end{theorem}

The geometry behind Theorem~\ref{thm:upper} is simple at the level of
dimensions. Given $r$ target syndromes, we look for $e-1$ locators
common to all of them and $e$ additional locators for each target.
This gives $re+e-1$ variables subject to $re$ power-sum equations, so
the expected dimension is $e-1$. The key observation is that the
corresponding projective variety has a smooth $F$-rational point at
infinity that is independent of the syndromes. The Jacobian criterion
singles out a unique irreducible component of the expected dimension; the
$q$-th power Frobenius fixes this component, showing that it is defined over
$F$, and a second Jacobian computation shows that it meets the affine chart.
B\'ezout's inequality and an explicit Lang--Weil estimate then produce
the required $F$-rational affine point.

The upper bound is complemented by two lower bounds, one coming from
the generalized supercode lemma and the Griesmer bound, and the other
from an elementary counting argument.

\begin{theorem}\label{thm:twosided}
Fix $e\ge2$ and $r\ge1$. For all sufficiently large $m$,
\[
\max\left\{
re,\,
\sum_{i=0}^{r-1}
\left\lceil\frac{2e-1}{2^i}\right\rceil
\right\}
\le
\rho_r\bigl(\BCH(e,m)\bigr)
\le
re+e-1.
\]
In particular,
\[
\rho_1\bigl(\BCH(e,m)\bigr)=2e-1
\qquad\text{and}\qquad
\rho_2\bigl(\BCH(e,m)\bigr)=3e-1
\]
for all sufficiently large $m$.
\end{theorem}

We next place these results in the context of previous work. The
ordinary covering radius of primitive binary BCH codes has a long
history. In particular, for every fixed $e$, Cohen \cite{Cohen} proved
that $\rho_1(\BCH(e,m))=2e-1$ for all sufficiently large $m$; Kavut and
Tutdere \cite{KavutTutdere} later improved the corresponding effective
range.

For higher generalized covering radii, several special cases have
recently been studied. For the double-error-correcting family ($e=2$),
Yohananov and Schwartz \cite{YohananovSchwartz} determined the second
generalized covering radius and also established the general lower bound
$\rho_2(\BCH(e,m))\ge3e-1$ in the relevant full-rank range.
\"{O}zbudak and \"{O}zt\"{u}rk \cite{OzbudakOzturkThird} proved that
$\rho_3(\BCH(2,m))=7$ for even $m\ge8$, while
$\rho_3(\BCH(2,m))\in\{6,7\}$ for odd $m\ge9$. Xiong and Yip
\cite{XiongYip}, combining combinatorial arguments with Weil-type
character-sum estimates, subsequently showed that, for every fixed $r$,
\[
2r\le \rho_r(\BCH(2,m))\le2r+1
\]
once $m$ is sufficiently large. These bounds agree with the
specialization of our general bounds to $e=2$, while the specialized
results give sharper information for particular extension degrees.

For the triple-error-correcting family ($e=3$), \"{O}zbudak and
\"{O}zt\"{u}rk \cite{OzbudakOzturkSecond} considered a related weak
second generalized covering radius and proved the upper bound $9$ for
odd $m\ge11$ and even $m\ge20$. Very recently, Essayag and
Zabokritskiy \cite{EssayagZabokritskiy} proved
\[
\rho_2\bigl(\BCH(3,m)\bigr)=8
\qquad (m\ge5).
\]
Their proof combines algebraic point counting with exact verification
of the remaining finite cases. Thus Theorem~\ref{thm:twosided}
recovers the same value for $e=3$ when
$m$ is sufficiently large and establishes $\rho_2(\BCH(e,m))=3e-1$ for
every fixed $e$ once $m$ is sufficiently large.

Most closely related to the present work is the very recent
completion-cover method of Belinsky and Zabokritskiy
\cite{BelinskyZabokritskiy}. For $r\ge2$, they obtain the same lower
bound as in Theorem~\ref{thm:twosided} once $m$ is sufficiently large,
and their
upper-bound construction starts from the same common-core parametrization
as ours, with $e-1$ locators shared by all targets and $e$ further
locators for each target. Their treatment of the resulting geometry uses
a technically sophisticated completion-cover framework. After first
removing targets that are already single BCH columns, each remaining
target gives a finite cover over the function field of the common
locators, with generic Galois group $S_e$, and the simultaneous problem
is reduced to linear disjointness of the corresponding completion
fields. In characteristic two, this is
controlled through quadratic sign subfields described by Artin--Schreier
classes, together with a matroid-intersection argument for the remaining
exceptional directions.

The delicate case consists of the pure highest-coordinate targets
$(0,\ldots,0,\gamma)$. The leading polar coefficient of the relevant
sign class is governed by $\gamma^{2e-3}$, and the map
$\gamma\mapsto\gamma^{2e-3}$ need not preserve $\F_2$-linear
independence. For $3\le e\le6$, Belinsky and Zabokritskiy compute
further terms of the corresponding Berlekamp classes and show that these
remove the possible cancellation; together with the straightforward case
$e=2$, this gives the bare common-core bound
$\rho_r(\BCH(e,m))\le(r+1)e-1$. As they emphasize,
``The cutoff $e=6$ marks only the range for which the exact sign-class
calculation is completed here'' \cite[Introduction]{BelinskyZabokritskiy};
in particular, they do not assert a failure of this bound for $e\ge7$.
For general $e$, if some pure basis directions fail the required
independence condition, they introduce a small common set of BCH columns
and shift the exceptional targets by suitable binary combinations of
these columns. After applying the common-core construction to the shifted
targets, the same columns are restored. This gives
\[
\rho_r\bigl(\BCH(e,m)\bigr)
\le (r+1)e-1+c_{e,m,r},
\]
where $c_{e,m,r}$ bounds the number of columns used in this
regularization step. It is explicit, vanishes when
$\gcd(2e-3,2^m-1)=1$, and is $O_e(1)$ uniformly in $r$.

Our argument starts from the same common-core system of power-sum equations
but bypasses this monodromy and linear-disjointness machinery. Rather than
constructing a component dominating the common-locator space, we exhibit
a single smooth $F$-rational point at infinity, independent of the
prescribed syndromes. The Jacobian criterion selects the unique
irreducible component through this point, the $q$-th power Frobenius shows
that it is defined over $F$, and a second Jacobian computation shows that
it meets the affine chart. B\'ezout's inequality and an explicit
Lang--Weil estimate then give
\[
\rho_r\bigl(\BCH(e,m)\bigr)\le(r+1)e-1
\]
for every fixed $e$ and $r$ once $m$ is sufficiently large. Thus, in this
asymptotic regime, the same common-core count follows from a substantially
simpler geometric argument. Their method, on the other hand, gives
sharper information in several finite-$m$ regimes and a better
multihomogeneous degree estimate, leading to better field-size
thresholds as $r$ grows. It would be interesting to obtain a comparable
degree estimate for the component selected by the point-at-infinity
argument.

For $r\ge3$ we do not claim that the upper bound is optimal; the exact
value for all sufficiently large $m$ remains open in general.

\medskip

The paper is organized as follows. Section~\ref{sec:lower} proves the
lower bounds. Section~\ref{sec:geometry} gives the geometric construction
and proves Theorem~\ref{thm:upper}. Section~\ref{sec:effective} records
explicit thresholds and consequences.

\section{Lower bounds}\label{sec:lower}

We first record two lower bounds.

Let $C_e=\BCH(e,m)$.
Then $C_e\subseteq C_{e-1}$.
In the full-rank range,
\[
\dim_{\F_2}C_{e-1}-\dim_{\F_2}C_e=m.
\]

\begin{proposition}\label{prop:griesmer}
Fix $r\ge1$ and assume $m\ge r$. Then
\[
\rho_r(C_e)
\ge
\sum_{i=0}^{r-1}
\left\lceil\frac{2e-1}{2^i}\right\rceil.
\]
\end{proposition}

\begin{proof}
Since
\[
\dim_{\F_2}C_{e-1}-\dim_{\F_2}C_e=m\ge r
\]
and $\dim_{\F_2}C_{e-1}\ge m\ge r$, the generalized supercode lemma
\cite[Lemma~III.1]{XiongYip} applies and gives
\[
\rho_r(C_e)\ge d_r(C_{e-1}),
\]
where $d_r$ denotes the $r$-th generalized Hamming weight.

The BCH bound gives $d(C_{e-1})\ge2e-1$.
Let $D\subseteq C_{e-1}$ be an arbitrary $r$-dimensional binary subcode.
Puncturing outside $\Supp(D)$ produces a binary linear code of length
$|\Supp(D)|$, dimension $r$, and minimum distance at least $2e-1$.
The binary Griesmer bound~\cite{Griesmer} therefore yields
\[
|\Supp(D)|
\ge
\sum_{i=0}^{r-1}
\left\lceil\frac{2e-1}{2^i}\right\rceil.
\]
Taking the minimum over all such $D$ proves the assertion.
\end{proof}

For $r=1$ this gives $2e-1$, while for $r=2$ it gives
\[
(2e-1)+
\left\lceil\frac{2e-1}{2}\right\rceil
=
3e-1.
\]

The following elementary counting argument gives a second lower bound with
leading term $re$.

\begin{proposition}\label{prop:counting}
Put $L=re-1$.
If $q\ge L$ and $q>\frac{2^{rL}}{L!}$,
then $\rho_r(C_e)\ge re$.
\end{proposition}

\begin{proof}
Suppose, to the contrary, that $\rho_r(C_e)\le L$.
By the definition of the generalized covering radius, this means that for every ordered $r$-tuple of syndromes
\[
(s^{(1)},\ldots,s^{(r)})\in(F^e)^r
\]
there exists a set of at most $L$ parity-check columns whose binary span
contains all the syndromes $s^{(1)},\ldots,s^{(r)}$.

We count how many ordered $r$-tuples can be covered in this way.

Fix a set $S$ of at most $L$ locators. Since $q\ge L$, we may enlarge $S$,
if necessary, to an $L$-element subset of $F$. This can only enlarge its
binary span. We also allow the auxiliary locator $0$; this introduces no new
syndrome and is used only to simplify the counting.
Thus it is enough to consider $L$-element subsets of $F$, of which there are
$\binom qL$.

Now fix one such set $S=\{x_1,\ldots,x_L\}$ and put
\[
W_S
=
\Span_{\F_2}\{h_e(x_1),\ldots,h_e(x_L)\}.
\]
Then $\dim_{\F_2}W_S\le L$, and hence $|W_S|\le2^L$.
Therefore the number of ordered $r$-tuples of syndromes all lying in
$W_S$ is at most $|W_S|^r\le 2^{rL}$.

If $\rho_r(C_e)\le L$, every ordered $r$-tuple of syndromes must belong to
$W_S^r$ for at least one $L$-element set $S$. Since the syndrome space is
$F^e$, it contains $q^e$ syndromes, and consequently there are $q^{er}$
ordered $r$-tuples altogether. By the union bound, $q^{er}
\le
\binom qL\,2^{rL}$.
Using $\binom qL\le q^L/L!$ and $er=L+1$, we obtain
$q^{L+1}
\le
\frac{q^L}{L!}2^{rL}$,
and therefore $q\le\frac{2^{rL}}{L!}$,
contradicting the hypothesis.
\end{proof}

\begin{corollary}\label{cor:lower}
For fixed $e$ and $r$, and all sufficiently large $m$,
\[
\rho_r(C_e)
\ge
\max\left\{
re,\,
\sum_{i=0}^{r-1}
\left\lceil\frac{2e-1}{2^i}\right\rceil
\right\}.
\]
\end{corollary}

\section{The geometric construction}\label{sec:geometry}

For an ordered $r$-tuple of syndromes
\[
s^{(1)},\ldots,s^{(r)}\in F^e,
\qquad
s^{(i)}
=
\bigl(
s^{(i)}_1,
s^{(i)}_3,
\ldots,
s^{(i)}_{2e-1}
\bigr),
\qquad
1\le i\le r,
\]
put
\[
\boldsymbol{s}:=(s^{(1)},\ldots,s^{(r)}).
\]
We shall look for $e-1$ locators common to all the syndromes and
$e$ additional locators for each individual syndrome. Thus there are
\[
(e-1)+re=re+e-1
\]
locators and $re$ simultaneous power-sum equations, so the expected
dimension is $e-1$.

We first record the following point-count estimate, which is a consequence
of the explicit Lang--Weil bound of Cafure--Matera~\cite{CafureMatera};
see also the classical theorem of Lang and Weil~\cite{LangWeil}.

\begin{lemma}[Cafure--Matera {\cite[Corollary~7.4]{CafureMatera}}]
\label{lem:CM}
Let $V\subseteq\A^n$ be an absolutely irreducible affine variety
over $\F_q$, of dimension $d>0$ and degree $\delta$. If
\[
q>\max\{2(d+1)\delta^2,\,2\delta^4\},
\]
then $V(\F_q)\ne\varnothing$.
\end{lemma}

The next proposition contains the geometric part of the argument. It identifies
an absolutely irreducible component of the simultaneous power-sum variety
with the properties needed below.

\begin{proposition}\label{prop:component}
Assume $q\ge e$. Let
\[
\mathcal X_{\boldsymbol{s}}\subseteq\PP^{re+e-1}
\]
have homogeneous coordinates
\[
T,\quad Z_1,\ldots,Z_{e-1},\quad
X_{i,1},\ldots,X_{i,e}
\qquad (1\le i\le r),
\]
and be defined by
\begin{equation}\label{eq:projective-incidence}
G_{i,j}
:=
\sum_{u=1}^{e-1}Z_u^{2j-1}
+\sum_{\ell=1}^{e}X_{i,\ell}^{2j-1}
-s^{(i)}_{2j-1}T^{2j-1}
=0
\end{equation}
for $1\le i\le r$ and $1\le j\le e$.
Let $H_\infty:=\{T=0\}$. Then
$\mathcal X_{\boldsymbol{s}}$ has an absolutely irreducible component
$Y_{\boldsymbol{s}}$, defined over $F$, such that
\[
\dim Y_{\boldsymbol{s}}=e-1,
\qquad
Y_{\boldsymbol{s}}\not\subseteq H_\infty,
\qquad
\deg Y_{\boldsymbol{s}}
\le
\left(\prod_{j=1}^e(2j-1)\right)^r.
\]
\end{proposition}

\begin{proof}
In the geometric arguments below, $\overline F$ denotes the algebraic closure of $F$, and all varieties and irreducible components are considered over $\overline F$ unless otherwise stated.

Since $q\ge e$, we may choose pairwise distinct elements $a_1=1,a_2,\ldots,a_{e-1}\in F^\times$ and put $a_e=0$. Consider the point $P\in\PP^{re+e-1}(F)$ given by
\[
T=0,\qquad
Z_u=a_u\quad(1\le u\le e-1),\qquad
X_{i,\ell}=a_\ell\quad(1\le i\le r,\ 1\le\ell\le e).
\]
For every $i$ and $j$,
\[
G_{i,j}(P)
=
2\sum_{u=1}^{e-1}a_u^{2j-1}
=
0
\]
since the characteristic is two. Hence $P\in \mathcal X_{\boldsymbol{s}}(F)\cap H_\infty$.
Since $Z_1(P)=1$, we work in the affine chart $Z_1\ne0$.

Fix $i$. At $P$, the $e\times e$ block of the Jacobian formed by
$G_{i,1},\ldots,G_{i,e}$ and the variables
$X_{i,1},\ldots,X_{i,e}$ is
\[
V
=
\bigl(a_\ell^{\,2j-2}\bigr)_{1\le j,\ell\le e}
=
\bigl((a_\ell^2)^{j-1}\bigr)_{1\le j,\ell\le e},
\]
since $\frac{d}{dX}X^{2j-1}=X^{2j-2}$ in characteristic two.
Thus $V$ is a Vandermonde matrix in
$a_1^2,\ldots,a_e^2$. These elements are pairwise distinct because
$a_1,\ldots,a_e$ are pairwise distinct and the squaring map $x\mapsto x^2$ is injective on $F$. Hence $V$ is
nonsingular.

For different values of $i$, these blocks involve disjoint sets of
variables. Hence the full Jacobian of the $re$ equations
\eqref{eq:projective-incidence} contains the block-diagonal minor $\operatorname{diag}(V,\ldots,V)$, with $r$ copies of $V$. Consequently its rank at $P$ is $re$. 

The affine chart $Z_1\ne0$ has dimension $re+e-1$, so the Jacobian criterion shows that $P$ is a smooth point of $\mathcal X_{\boldsymbol{s}}$ of codimension $re$. Hence $P$ lies on a unique irreducible component $Y$, and
\[
\dim Y=(re+e-1)-re=e-1.
\]

Since $P$ is $F$-rational, the $q$-th power Frobenius map on $\PP^{re+e-1}(\overline F)$,
\[
\Frob_q: [x_0:\cdots:x_{re+e-1}] \longmapsto [x_0^q:\cdots:x_{re+e-1}^q],
\]
fixes $P$ and permutes the irreducible components of $\mathcal X_{\boldsymbol{s}}$. Hence $\Frob_q(Y)$ is again an irreducible component containing $P$. By uniqueness, $\Frob_q(Y)=Y$. Since $Y$ is invariant under the $q$-th power Frobenius, it is defined over $F$. Denote this component by $Y_{\boldsymbol{s}}$. It is absolutely irreducible, and $\dim Y_{\boldsymbol{s}}=e-1$.

It remains to show that $Y_{\boldsymbol{s}}\not\subseteq H_\infty$. Consider the variety $\mathcal X_{\boldsymbol{s}}\cap H_\infty$, obtained by adding the equation $T=0$ to \eqref{eq:projective-incidence}. At $P$, take the same $re$ columns corresponding to the variables $X_{i,\ell}$, together with the $T$-column. The resulting $(re+1)\times(re+1)$ Jacobian minor has the form
\[
\begin{pmatrix}
\operatorname{diag}(V,\ldots,V) & *\\
0 & 1
\end{pmatrix},
\]
and is therefore nonsingular. Hence $P$ is a smooth point of
$\mathcal X_{\boldsymbol{s}}\cap H_\infty$, of dimension
\[
(re+e-1)-(re+1)=e-2.
\]
If $Y_{\boldsymbol{s}}\subseteq H_\infty$, then
$Y_{\boldsymbol{s}}$ would be an irreducible subvariety of
$\mathcal X_{\boldsymbol{s}}\cap H_\infty$ through $P$ of dimension
$e-1$, a contradiction. Thus
$Y_{\boldsymbol{s}}\not\subseteq H_\infty$.

Finally, $Y_{\boldsymbol{s}}$ is an irreducible component of the common
zero locus of the $re$ hypersurfaces \eqref{eq:projective-incidence}.
For each fixed $i$, the $e$ hypersurfaces indexed by $j$ have degrees
$1,3,\ldots,2e-1$. Hence B\'ezout's inequality gives
\[
\deg Y_{\boldsymbol{s}}
\le
\prod_{i=1}^r\prod_{j=1}^e(2j-1)
=
\left(\prod_{j=1}^e(2j-1)\right)^r. \qedhere
\]
\end{proof}

We now pass to the affine incidence variety. Let
$X_{\boldsymbol{s}}\subseteq\A^{re+e-1}$ be defined by
\begin{equation}\label{eq:affine-incidence}
\sum_{u=1}^{e-1}z_u^{2j-1}
+
\sum_{\ell=1}^{e}x_{i,\ell}^{2j-1}
=
s^{(i)}_{2j-1},
\qquad
1\le i\le r,\quad 1\le j\le e.
\end{equation}
This is precisely the affine chart $T\ne0$ of
$\mathcal X_{\boldsymbol{s}}$, after normalizing $T=1$.

Put
\begin{equation}\label{eq:Delta}
\Delta_{e,r}
:=
\left(\prod_{j=1}^e(2j-1)\right)^r
=
\bigl((2e-1)!!\bigr)^r.
\end{equation}

\begin{proof}[Proof of Theorem~\ref{thm:upper}]
Take $m$ sufficiently large that the usual BCH full-rank conditions
hold, that $q=2^m\ge e$, and that
\begin{equation}\label{eq:LW-condition-main}
q>2\Delta_{e,r}^{\,4}.
\end{equation}
Fix an arbitrary ordered $r$-tuple $\boldsymbol{s}$ of syndromes as
above.

By Proposition~\ref{prop:component},
$Y_{\boldsymbol{s}}\not\subseteq H_\infty$, so its affine part
\[
V_{\boldsymbol{s}}
:=
Y_{\boldsymbol{s}}\cap\{T\ne0\}
\]
is nonempty and absolutely irreducible, defined over $F$, and has
dimension $e-1$. Its projective closure is $Y_{\boldsymbol{s}}$, so
\[
\deg V_{\boldsymbol{s}}
=
\deg Y_{\boldsymbol{s}}
\le
\Delta_{e,r}.
\]

Lemma~\ref{lem:CM} therefore applies as soon as
\[
q>
\max\left\{
2e\Delta_{e,r}^{\,2},
2\Delta_{e,r}^{\,4}
\right\}.
\]
Since $\Delta_{e,r}\ge 2e-1$, we have
$\Delta_{e,r}^{\,2}>e$, and hence
\[
2e\Delta_{e,r}^{\,2}
<
2\Delta_{e,r}^{\,4}.
\]
Thus \eqref{eq:LW-condition-main} guarantees that
$V_{\boldsymbol{s}}(F)\ne\varnothing$.

Choose an $F$-rational point
\[
z_1,\ldots,z_{e-1},
\qquad
x_{i,1},\ldots,x_{i,e}
\qquad (1\le i\le r)
\]
of $V_{\boldsymbol{s}}$. Equation \eqref{eq:affine-incidence} gives
\[
s^{(i)}
=
\sum_{u=1}^{e-1}h_e(z_u)
+
\sum_{\ell=1}^{e}h_e(x_{i,\ell}),
\qquad
1\le i\le r.
\]

Thus all the prescribed syndromes lie in the $\F_2$-span of the
columns indexed by the distinct nonzero elements among the displayed
$(e-1)+re$ locators. There are at most
\[
(e-1)+re=re+e-1
\]
such columns. Since $\boldsymbol{s}$ was arbitrary,
\[
\rho_r\bigl(\BCH(e,m)\bigr)
\le
re+e-1.\qedhere
\]
\end{proof}

\begin{remark}\label{rem:higher-r}
For $r\ge3$, the construction need not be optimal. A column in a
common support for $r$ syndromes can occur with any nonzero coefficient
vector in $\F_2^r$. Our parametrization uses $e-1$ locators with the
all-ones coefficient vector and, for each $1\le i\le r$, $e$ locators
with coefficient vector equal to the $i$-th standard basis vector of
$\F_2^r$. When $r=2$, these are exactly the three nonzero vectors in
$\F_2^2$, and the resulting upper bound is sharp for all sufficiently
large $m$.

When $r\ge3$, there are additional nonzero coefficient vectors. They
can arise when some of the chosen locators coincide, since the
corresponding coefficient vectors then add in $\F_2^r$, but they are
not present as separate variables in our parametrization. Allowing
them from the outset could lead to smaller common supports. Thus the
present argument does not determine, for a given $e$ and $r\ge3$,
whether
\[
\rho_r\bigl(\BCH(e,m)\bigr)
=
(r+1)e-1
\]
for all sufficiently large $m$.
\end{remark}

\section{Explicit thresholds and consequences}
\label{sec:effective}

The proof of Theorem~\ref{thm:upper} already yields a simple explicit
field-size threshold. Put
\[
D_e
:=
\prod_{j=1}^{e}(2j-1)
=
(2e-1)!!.
\]
Since
\[
\Delta_{e,r}=D_e^r,
\]
condition \eqref{eq:LW-condition-main} becomes
\begin{equation}\label{eq:q-upper}
q=2^m>2D_e^{4r}.
\end{equation}

We therefore obtain the following effective form of the upper bound.

\begin{proposition}\label{prop:explicit-upper}
Assume that the usual BCH full-rank condition
\begin{equation}\label{eq:BCH-rank-effective}
2e-1\le 2^{\lceil m/2\rceil}
\end{equation}
holds. If
\[
2^m>2D_e^{4r},
\]
then
\[
\rho_r\bigl(\BCH(e,m)\bigr)
\le
(r+1)e-1.
\]
\end{proposition}

\begin{proof}
The BCH rank condition gives the parity-check description
\eqref{eq:column}. The field-size inequality is precisely the condition
used in the proof of Theorem~\ref{thm:upper}. Hence the conclusion follows
directly from that theorem and its proof.
\end{proof}

We now translate these conditions into explicit bounds on $m$.
Define
\begin{equation}\label{eq:mBCH}
m_{\mathrm{BCH}}(e)
:=
2\left\lceil\log_2(2e-1)\right\rceil-1.
\end{equation}
Then every $m\ge m_{\mathrm{BCH}}(e)$ satisfies
\eqref{eq:BCH-rank-effective}.

Similarly, put
\begin{equation}\label{eq:mLW}
m_{\mathrm{LW}}(e,r)
:=
\min\left\{
m\ge1:
2^m>2D_e^{4r}
\right\}.
\end{equation}
Equivalently,
\[
m_{\mathrm{LW}}(e,r)
=
\left\lfloor
1+4r\log_2 D_e
\right\rfloor+1.
\]

\begin{corollary}\label{cor:explicit-upper-m}
Put
\begin{equation}\label{eq:mup}
m_{\mathrm{up}}(e,r)
:=
\max\left\{
m_{\mathrm{BCH}}(e),
m_{\mathrm{LW}}(e,r)
\right\}.
\end{equation}
Then, for every $m\ge m_{\mathrm{up}}(e,r)$,
\[
\rho_r\bigl(\BCH(e,m)\bigr)
\le
(r+1)e-1.
\]
\end{corollary}

Thus, for fixed $e$, the extension degree supplied by the present argument
grows linearly with $r$:
\[
m_{\mathrm{up}}(e,r)
=
4r\log_2D_e+O_e(1).
\]

We can make the lower-bound range explicit as well. Remember that $L=er-1$.
Proposition~\ref{prop:counting} applies as soon as $2^m\ge L$ and $2^m>\frac{2^{rL}}{L!}$.
Accordingly, define
\begin{equation}\label{eq:mcount}
m_{\mathrm{count}}(e,r)
:=
\max\left\{
\left\lceil\log_2L\right\rceil,\,
\left\lfloor
rL-\log_2(L!)
\right\rfloor+1
\right\}.
\end{equation}
Together with the condition $m\ge r$ needed in
Proposition~\ref{prop:griesmer}, this gives the following explicit
two-sided version of Theorem~\ref{thm:twosided}.

\begin{corollary}\label{cor:explicit-two-sided}
Let
\begin{equation}\label{eq:mtwo}
m_{\mathrm{two}}(e,r)
:=
\max\left\{
m_{\mathrm{up}}(e,r),\,
r,\,
m_{\mathrm{count}}(e,r)
\right\}.
\end{equation}
Then, for every $m\ge m_{\mathrm{two}}(e,r)$,
\[
\max\left\{
re,\,
\sum_{i=0}^{r-1}
\left\lceil\frac{2e-1}{2^i}\right\rceil
\right\}
\le
\rho_r\bigl(\BCH(e,m)\bigr)
\le
(r+1)e-1.
\]
\end{corollary}

For the first two generalized covering radii, the counting lower bound is
not needed: the Griesmer lower bound already matches the upper bound.
Consequently the effective upper threshold itself gives exact values.

\begin{corollary}\label{cor:effective-r1-r2}
For every fixed $e\ge2$,
\[
\rho_1\bigl(\BCH(e,m)\bigr)=2e-1
\qquad
\text{for all }m\ge m_{\mathrm{up}}(e,1),
\]
and
\[
\rho_2\bigl(\BCH(e,m)\bigr)=3e-1
\qquad
\text{for all }m\ge m_{\mathrm{up}}(e,2).
\]
\end{corollary}

The thresholds furnished by the general argument for the second radius are
shown below. They are sufficient thresholds only and are not intended to
be minimal.

\begin{center}
\begin{tabular}{c|r|c|c}
\toprule
$e$ & $D_e=(2e-1)!!$
& $m_{\mathrm{up}}(e,2)$
& $\rho_2(\BCH(e,m))$\\
\midrule
$2$ & $3$       & $14$  & $5$\\
$3$ & $15$      & $33$  & $8$\\
$4$ & $105$     & $55$  & $11$\\
$5$ & $945$     & $81$  & $14$\\
$6$ & $10395$   & $108$ & $17$\\
$7$ & $135135$  & $138$ & $20$\\
$8$ & $2027025$ & $169$ & $23$\\
\bottomrule
\end{tabular}
\end{center}

For $r=3$, the same calculation gives, for example,
\[
m_{\mathrm{up}}(3,3)=48,\qquad
m_{\mathrm{up}}(4,3)=82,\qquad
m_{\mathrm{up}}(5,3)=120,\qquad
m_{\mathrm{up}}(6,3)=162.
\]
Combining these upper thresholds with the Griesmer bound yields
\[
10\le
\rho_3\bigl(\BCH(3,m)\bigr)
\le11
\qquad (m\ge48),
\]
\[
13\le
\rho_3\bigl(\BCH(4,m)\bigr)
\le15
\qquad (m\ge82),
\]
\[
17\le
\rho_3\bigl(\BCH(5,m)\bigr)
\le19
\qquad (m\ge120),
\]
and
\[
20\le
\rho_3\bigl(\BCH(6,m)\bigr)
\le23
\qquad (m\ge162).
\]

\section*{Acknowledgments}
The last author is very grateful for the hospitality of the Department of
Mathematics of The Hong Kong University of Science and Technology, where
he spent two weeks as a visiting researcher.
The last author was partially supported by the Italian National Group for
Algebraic and Geometric Structures and their Applications
(GNSAGA--INdAM).

\end{document}